\documentclass{amsart}

\usepackage[latin1]{inputenc}
\usepackage[all]{xy}
\usepackage{amssymb, amsmath, amscd, geometry}
\usepackage{graphicx}
\usepackage{setspace}

\numberwithin{equation}{section}
\input xy
\xyoption{all}
\usepackage{latexsym}
\usepackage[usenames]{color}
\usepackage{epic} %%%% Circunferencias grandes y curvas de bezier
\usepackage{mathrsfs}
\usepackage{euscript}
\usepackage{rotating}
\usepackage{verbatim}
\newtheorem{lemma}{Lemma}[section]
\newtheorem{theorem}[lemma]{Theorem}

\newtheorem{prop}[lemma]{Proposition}

\newtheorem*{theorem*}{Theorem}

\theoremstyle{definition}

      \newcommand{\N}{{\mathbb N}}
      
      \newcommand{\C}{{\mathbb C}}
\newcommand{\gM}{\mathsf{M}}

\newcommand{\Tk}{S_1^{k,\mathrm{sa}}}
\newcommand{\Herm}[1]{\gM_{#1}^{\mathrm{sa}}}

\newcommand{\tr}{\operatorname{tr}}

\newcommand{\id}{\operatorname{id}}

\newcommand{\M}{\mathcal M}
\newcommand{\RR}{\mathcal R}

\newcommand*{\coloneqq}{\mathrel{\vcenter{\baselineskip0.5ex \lineskiplimit0pt \hbox{\scriptsize.}\hbox{\scriptsize.}}} =}
\newcommand{\uno}{1\!\!1}

\newcommand{\cH}{\mathcal{H}}

\newcommand{\lo}{\mathrm{LO}}
\newcommand{\loccone}{\mathrm{LOCC}_\to}
\newcommand{\locc}{\mathrm{LOCC}}
\newcommand{\sep}{\mathrm{SEP}}
\newcommand{\all}{\mathrm{ALL}}

\title[Maximal gap between local and global distinguishability of bipartite quantum states]{Maximal gap between local and global distinguishability of bipartite quantum states}

\author{Willian H. G. Corr\^ea}
\address{Departamento de Matem\'atica, Instituto de Ci\^{e}ncias Matem\'aticas e de Computa\c{c}\~{a}o, Universida de S\~{a}o Paulo, Avenida Trabalhador S\~{a}o-carlense, 400-Centro CEP: 13566-590 - S\~{a}o Carlos- SP, Brazil}
\email{willhans@icmc.usp.br}

\author{Ludovico Lami}
\address{Institute of Theoretical Physics and IQST, Universit\"{a}t Ulm, Albert-Einstein-Allee 11D-89069 Ulm, Germany}
\email{ludovico.lami@gmail.com}

\author{Carlos Palazuelos}
\address{Dpto. An\'alisis Matem\'atico y Matem\'atica Aplicada, Fac. Ciencias Matemáticas, Universidad Complutense de Madrid, Plaza de Ciencias s/n 28040 Madrid, Spain \& Instituto de Ciencias Matem\'aticas, C/ Nicol\'as Cabrera, 13-15, 28049 Madrid, Spain}
\email{carlospalazuelos@mat.ucm.es}

\addtocounter{tocdepth}{-1}

\begin{document}

\keywords{}

\maketitle

\begin{abstract}
We prove a tight and close-to-optimal lower bound on the effectiveness of local quantum measurements (without classical communication) at discriminating any two bipartite quantum states. Our result implies, for example, that any two orthogonal quantum states of a $n_A\times n_B$ bipartite quantum system can be discriminated via local measurements with an error probability no larger than $\frac12 \left(1 - \frac{1}{c \min\{n_A, n_B\}} \right)$, where $1\leq c\leq 2\sqrt2$ is a universal constant, and our bound scales provably optimally with the local dimensions $n_A,n_B$. Mathematically, this is achieved by showing that the distinguishability norm $\|\cdot\|_{\lo}$ associated with local measurements satisfies that $\|\cdot\|_1\leq 2\sqrt2 \min\{n_A,n_B\} \|\cdot\|_{\lo}$, where $\|\cdot\|_1$ is the trace norm.
\end{abstract}

\section{Distinguishability of bipartite quantum states}\label{sec:Distinguishability}

The existence of pairs of bipartite quantum states that can be perfectly distinguished via general entangled measurements while being almost indistinguishable under local operations and classical communication (LOCC) is usually referred to as `quantum data hiding'. It is a distinctive property of bipartite quantum systems~\cite{TDL01, DLT02}. More generally, if a set of measurements $\mathcal M$ on a bipartite quantum system is fixed, one can study the effectiveness of discriminating protocols to distinguish two states $\rho$ and $\sigma$ with a priori probabilities $p$ and $1-p$ respectively, when only operations from the set $\mathcal M$ are available. A natural way to study this problem is by minimizing the corresponding probability of error. This can be done by introducing the distinguishability norm associated with $\mathcal M$~\cite{MWW09}, defined as
\begin{equation}\label{dist_norm}
\|x\|_{\mathcal M} \coloneqq \sup_{(e_i)_{i\in \mathcal I}\in \mathcal M}\sum_{i}|\tr(e_ix)|.
\end{equation}
By means of the above quantity, we can express the probability of error in discrimination by means of measurements in $\mathcal{M}$ as 
$$P_e^{\mathcal M}(\rho, \sigma, p)\coloneqq \frac{1}{2}(1-\|p\rho-(1-p)\sigma\|_{\mathcal M}) .$$
We thus see that the worst-case efficiency of the set $\M$ at binary state discrimination is effectively quantified by the data hiding ratio
$$\mathcal R(\mathcal M)= \sup_{\rho,\sigma,p} \frac{\|p\rho-(1-p)\sigma\|_{\all}}{\|p\rho-(1-p)\sigma\|_{\mathcal M}},$$
where $\|\cdot\|_{\all}$ denotes the norm associated with the whole set of possible measurements and the maximization runs over all pairs of states $\rho$, $\sigma$ and all $p\in (0,1)$.

Although the original formulation of the data hiding problem was stated for $\M=LOCC$, many other sets of measurements with equally legitimate operational interpretations have been considered in the recent literature. The set of local operations (LO), the set of local operations assisted by one-way classical communication (LOCC$_{\rightarrow}$) and the set of separable measurements (SEP) are some notable examples of families of measurements for which $\mathcal R(\mathcal M)$ has been deeply studied, with the motivation that it provides a precise quantification of the operational power of these families of measurements. In what follows, let us consider a bipartite quantum system $AB$ with finite-dimensional Hilbert space $\cH_{AB}=\cH_A\otimes \cH_B$. We will denote with $n_A\coloneqq \dim\cH_A$ (resp.\ $n_B\coloneqq \dim \cH_B$) the local dimension, and with $\mathbf{M}_{A}$ (resp.\ $\mathbf{M}_{B}$) the set of all measurements on system $A$ (resp.\ $B$). Let us define
\begin{align*}
\lo\, &\coloneqq\, \left\{ (e_{i}\otimes f_{j})_{(i,j)\in I\times J}:\ (e_{i})_{i\in I}\in\mathbf{M}_{A},\ (f_{j})_{j\in J}\in\mathbf{M}_{B} \right\},% \\\text{\emph{LOSR}}\, &\coloneqq\, \left\{ \big(p_{i}\, e^{(i)}_{j}\otimes f^{(i)}_{k}\big)_{(i,j,k)\in I\times J\times K}:\ I\ \text{finite},\ p\in\mathcal{P}(I),\ \big(e^{(i)}_{j}\big)_{j\in J}\in\mathbf{M}_{A},\ \big(f_{k}^{(i)}\big)_{k\in K}\in\mathbf{M}_{B}\ \forall\ i\in I \right\} , 
\\[1ex]\loccone\, &\coloneqq\, \left\{ (e_{i}\otimes f_{j}^{(i)})_{(i,j)\in I\times J}:\ (e_{i})_{i\in I}\in\mathbf{M}_{A},\ (f_{j}^{(i)})_{j\in J}\in\mathbf{M}_{B}\ \forall\ i\in I \right\} , \\
\sep\, &\coloneqq\, \left\{ (e_{i}\otimes f_{i})_{i\in I} \in \mathbf{M}_{AB} \right\}.
\end{align*}

It is straightforward to verify the chain of inclusions $\lo\subseteq \loccone \subseteq \locc \subseteq \mathrm{SEP}$, which easily implies the following chain of inequalities $\mathcal R(\sep)\leq \mathcal R(\locc)\leq \mathcal R(\loccone)\leq \mathcal R(\lo)$. Moreover, sharp estimates have been obtained for some of these quantities since it is known that~\cite[Theorem~16]{LPW18}
\begin{align}\label{optimal est SEP-LOCC}
\min\{n_A, n_B\}\leq \mathcal R(\sep)\leq \mathcal R(\locc)\leq 2\min\{n_A, n_B\}-1
\end{align}
and moreover that
\begin{equation} \label{optimal est loccone}
\mathcal R(\loccone)\leq 2n_A-1.
\end{equation}
The asymmetry of the above inequality with respect to the exchange of the two subsystems $A,B$ is due to the fact that the set of measurements $\loccone$ is itself asymmetric, featuring a one-way communication from Alice to Bob but not vice versa. If Alice's system is the smaller of the two, then the right-hand sides of~\eqref{optimal est SEP-LOCC} and~\eqref{optimal est loccone} coincide.

Equation~\eqref{optimal est SEP-LOCC} shows, in particular, that $\mathcal R(\locc)$ scales linearly with the minimum local dimension. More interestingly, it implies the a priori nontrivial fact that extending the set of LOCC operations to the larger set SEP or reducing it to the smaller set $\loccone$ (with the communication being from the smaller to the larger system) does not modify the scaling of the data hiding ratio. In this context, the problem of determining the scaling of $\mathcal R(\lo)$ has so far remained open. The best known upper bounds prior to our work read
\begin{align}\label{Previous upper bound LO}
&\mathcal R(\lo)\leq \sqrt{153\, n_An_B},\qquad \mathcal R(\lo)\leq 4\min\left\{n_A^{3/2}, n_B^{3/2}\right\}.
\end{align}

The first estimate in Equation~\eqref{Previous upper bound LO} was proved in~\cite[Theorem 15]{MWW09} as an application of Berger's inequality for random variables. The second upper bound in Equation~\eqref{Previous upper bound LO} was shown instead in~\cite[Corollary 9]{ALPSW20} and it is based on a tight estimate for the quotient of the $\pi$ over the $\epsilon$ tensor norms on a certain tensor product of Banach spaces. Note that the two bounds in~\eqref{Previous upper bound LO} are incomparable: the tighter of the two is the former e.g.\ for sufficiently large $n_A=n_B$, and the latter for fixed $n_A$ and large $n_B$.

Since any reasonable set of measurements should contain the set of local operations, the quantity $\mathcal R(\lo)$ can be understood as the ultimate upper bound for data hiding, in the sense that it provides an upper bound for the quantity $\RR(\M)$ for any reasonable set of measurements $\M$. On the other hand, LO operations are the natural ones in those scenarios where Alice and Bob have local quantum memories whose coherence time is much shorter than the time light takes to travel between them, so that the exchange of classical messages is not an option for them.

Our main result is as follows.

\begin{theorem}\label{main hidden}
For bipartite quantum systems with local dimensions $n_A$ and $n_B$, the following inequality holds:
$$\RR(\lo)\leq 2\sqrt{2}\min\{n_A, n_B\}.$$
\end{theorem}

This result establishes the optimal scaling of $\RR(\lo)$ with respect to the local dimensions. It shows that the data hiding ratios associated with LO and LOCC are of the same order, something (arguably) unexpected. Indeed, although there exist pairs of states that are much more efficiently discriminated via LOCC than via LO~\cite[Theorem~3]{Aubrun-2015}, the largest possible gap against global measurements is the same for both sets. Moreover, our result implies that the gap remains the same if one further restricts the local measurements to be binary, and the a posteriori communication to a single bit (in fact, even an XOR discrimination scheme suffices). In fact, according to Eq.~\eqref{optimal est SEP-LOCC}, all different sets of measurements give rise to data hiding ratios scaling as $\Theta\left(\min\{n_A,n_B\}\right)$.  

As we will see, Theorem~\ref{main hidden} follows from an optimal estimate between two norms in a certain Banach space (see Proposition~\ref{key prop} below). Once the problem is reduced to this setting, the main ingredient of our proof is the so called non-commutative Grothendieck's theorem for general C$^*$-algebras.

As an application illustrating the fundamental status of our result, we give improved bounds on the phenomenon of quantum Darwinism~\cite{Zurek2009, Brandao2015, Knott2018, Eugenia-Darwin, Ranard2020}. In spite of the slightly confusing name, this refers to the emergence of objective features when a single quantum system is probed by many other systems that play the role of observers. Our contribution is to substantially better the scaling of the dimensional coefficients governing the rate at which the transition between microscopic (quantum) and macroscopic (classical) world takes place. Prior to our work, the best coefficient has been reported by Qi and Ranard~\cite[Eq.~(8) and~(31)]{Ranard2020}, and reads $\Omega_{d_A, d_R} = \min\left\{ d_A^2, 4d_A^{3/2}, 4d_R^{3/2}, \sqrt{153 d_A d_R}, 2d_R-1 \right\}$; we show below how this can be improved to $\Omega_{d_A,d_R}=\min\left\{ 2\sqrt2 d_A,\, 2d_R-1\right\}$. This is a substantial improvement because it reduces the scaling in $d_A$ from $O(d_A^{3/2})$ to $O(d_A)$. Note that $d_A$, which is an intrinsic feature of the observed system and independent of the observer, is the most important parameter here, to the point that some authors consider $d_R$ unconstrained, potentially unbounded and even infinite~\cite{Brandao2015} --- for example, $d_R$ will be astronomically large when $R$ contains a human observer. Thus, often the only practically relevant parameter is $d_A$, which can be relatively small if the observed system contains a few qubits or atoms. Our result thus implies that the classical regime is entered earlier than previously expected, i.e.\ for a smaller number of observing systems.

\section{Proof of the main result}
In order to prove Theorem~\ref{main hidden} we need to introduce some basic elements from Banach space theory. Let us denote by $\gM_k$ (resp.\ $\Herm{k}$) the complex (resp.\ real) vector space of $k\times k$ (resp.\ selfadjoint) matrices. We will consider here the trace norm $\|\cdot\|_1$  and the operator norm $\|\cdot\|_\infty$ and denote the corresponding spaces  $S_1^{k}$ (resp.\ $\Tk$) and  $S_\infty^{k}$ (resp.\ $S_\infty^{k,\mathrm{sa}}$). It is well known that these spaces are dual to each other: $(S_1^{k})^*=S_\infty^{k}$, $(S_\infty^{k})^*=S_1^{k}$ isometrically (and similarly for the selfadjoint versions), where the duality is given by the Hilbert--Schmidt inner product, $\langle x,  y \rangle= \tr(xy)$. Remember that $\gM_n\otimes \gM_m=\gM_{nm}$ and $\Herm{n}\otimes \Herm{m}=\Herm{nm}$ canonically. 

We will also use some standard identifications between bilinear forms and tensor products. In particular, given two (real or complex) normed spaces $X$ and $Y$ and an element $z\in X\otimes Y$, we define the $\epsilon$ tensor norm by
\begin{equation}\label{inj_norm}
\|z\|_{X\otimes_{\epsilon} Y}=\sup\{(x^*\otimes y^*)(z)\}, 
\end{equation}where the supremum runs over all elements $x^*$  and $y^*$ in the dual space $X^*$ and $Y^*$ respectively, such that $\|x\|_{X^*}\leq 1$, $\|y\|_{Y^*}\leq 1$. We will denote by $X\otimes_{\epsilon} Y$ the algebraic tensor product $X\otimes Y$ endowed with the $\epsilon$ tensor norm.

Now, if $X$ and $Y$ are finite dimensional spaces the space of bilinear forms on $X\times Y$, can be canonically identified with the tensor product $X^*\otimes Y^*$ (the trivial direction assigns, to any elementary tensor $x^*\otimes y^*\in X^*\times Y^*$, the bilinear form $B:X\times Y\rightarrow \C$ defined as $B(x,y)=x^*(x)y^*(y)$). Moreover, if we consider the standard norm for bilinear forms and denote the corresponding normed space by  $\mathcal B(X\times Y)$, the following identification is isometric:
\begin{equation}\label{duality_Bil_tensor}
\mathcal B(X\times Y)=X^*\otimes_\epsilon Y^*.
\end{equation}

The proof of Theorem~\ref{main hidden} is based on reducing its statement to the comparison of two well known norms in Banach spaces. To this end, let us look at the distinguishability norm a little bit more formally than before. First note that, given  bipartite quantum states $\rho$ and $\sigma$ of local dimension $n_A$ and $n_B$ and $p\in [0,1]$, the element $z=p\rho-(1-p)\sigma$ naturally lives in $\Herm{n_An_B}$. Then, one can show that if a set of measurements $\M$ is informationally complete\footnote{$\M$ is informationally complete if $\mathrm{span}\left\{e_i: i\in I;\, \Lambda=(e_i)_{i\in I}\in \M\right\}=\Herm{n_An_B}$. Since this property is satisfied by the sets considered in this work, $\M=\all$ and $\M=\lo$, we will always assume $\M$ to be informationally complete.} Equation~\eqref{dist_norm} defines a norm on $\Herm{n_An_B}$. Now, given a set of measurements $\M$, let us denote by $\langle \M\rangle$ the set generated by $\M$ via \emph{coarse graining}; mathematically, 
\begin{equation}
\langle \mathcal{M} \rangle\, \coloneqq\, \bigg\{ (e_{j})_{j\in J}:\ \exists\ I\ \text{finite},\ \{I_{j}\}_{j\in J}\ \text{partition of}\ I:\ (e_{i})_{i\in I}\in\mathcal{M},\ e_{j}=\sum_{i\in I_{j}} e_{i}\ \, \forall\, j\in J \bigg\} .
\label{coarse}
\end{equation} 

The following reformulation of~\cite[Lemma 2]{MWW09} gives us a very useful description of $\|\cdot\|_{\M}$.
\begin{lemma}\label{description dist_norm}
Let $\M$ be a set of measurements and $z\in \Herm{n_An_B}$. Then,
\begin{equation}\label{Eq. description dist_norm}
\|z\|_{\M}=\sup\left\{\tr(zX): X\in \Herm{n_An_B},\,  \left(\frac{\uno_{AB}+ X}{2}, \frac{\uno_{AB}-X}{2}\right)\in \langle \mathcal{M} \rangle\right\}.
\end{equation}
\end{lemma}

The following lemma will allow us to reduce our problem on $\RR(\lo)$ to a problem about tensor norms.
\begin{lemma}\label{lemma norms}
Given two bipartite quantum states $\rho$ and $\sigma$ of local dimension $n_A$ and $n_B$, $p\in [0,1]$ and $z=p\rho-(1-p)\sigma\in \Herm{n_An_B}$, we have
\begin{enumerate}
    \item $\|z\|_{\all}=\|z\|_{S_1^{n_An_B,\mathrm{sa}}}$.
    \item $\|z\|_{S_1^{n_A,\mathrm{sa}}\otimes_\epsilon S_1^{n_B,\mathrm{sa}}}\leq \|z\|_{\lo}$.
\end{enumerate}
\end{lemma}

\begin{proof}
In order to prove item (1) we just need to apply Lemma~\ref{description dist_norm} to the set $\M=\all$. Then, Lemma~\ref{description dist_norm} guarantees that 
$$\|z\|_{\all}=\sup\left\{\tr(zX): X\in \Herm{n_An_B},\,  \left(\frac{\uno_{AB}+X}{2}, \frac{\uno_{AB}-X}{2}\right)\in \all \right\}.$$ Now, condition $\left(\frac{\uno_{AB}+X}{2}, \frac{\uno_{AB}-X}{2}\right)\in \all$ is equivalent to the fact that $\uno_{AB}\pm X$ are both semidefinite positive matrices. This is clearly equivalent to $X$ being in the unit ball of $S_\infty^{n_An_B,\mathrm{sa}}$. Hence, the equality in item (1) follows from the duality relation $S_\infty^{n_An_B,\mathrm{sa}}=(S_1^{n_An_B,\mathrm{sa}})^*$.

Item 2 was actually proved in~\cite[Proposition 22]{LPW18} in much more generality. We present a self-contained proof here for completeness. According to the definition of the injective tensor norm~\eqref{inj_norm}, given elements $f$ and $g$ in the unit ball of $S_\infty^{n_A,\mathrm{sa}}$ and $S_\infty^{n_B,\mathrm{sa}}$ respectively, we must show that $\tr(z\, f\otimes g)\leq \|z\|_{\lo}$. To this end, we will show that $A=f\otimes g$ is one of the elements appearing in Equation~\eqref{Eq. description dist_norm} when $\M=\lo$. Indeed, this follows from the trivial identities:
\begin{align*}
\frac12 \left( \uno_A\otimes \uno_B + f\otimes g \right)\, &=\, \frac{\uno_A+f}{2}\otimes \frac{\uno_B+g}{2}\, +\, \frac{\uno_A-f}{2}\otimes \frac{\uno_B-g}{2}\, ,  \\
\frac12 \left( \uno_A\otimes \uno_B - f\otimes g \right)\, &=\, \frac{\uno_A+f}{2}\otimes \frac{\uno_B-g}{2}\, +\, \frac{\uno_A-f}{2}\otimes \frac{\uno_B+g}{2}\, , 
\end{align*}
joint with the fact that $\left\{ \frac12 (\uno_A+ f),\, \frac12(\uno_A-f) \right\}$ is a valid measurement on $A$, and analogously $\left\{ \frac12 (\uno_B+ g),\, \frac12(\uno_B-g) \right\}$ is a measurement on $B$. This easily implies that $$\left\{ \frac12 (\uno_A\otimes \uno_B + f\otimes g),\, \frac12 (\uno_A\otimes\uno_B - f\otimes g)\right\} \in \langle\lo\rangle.$$
\end{proof}

According to Lemma~\ref{lemma norms}, Theorem~\ref{main hidden} follows from the next proposition, which is the key result of the paper.
\begin{prop}\label{key prop}
For any $z\in \gM_{n_An_B}$ we have that $$\|z\|_{S_1^{n_An_B}}\leq 2\min\{n_A,n_B\}\|z\|_{S_1^{n_A}\otimes_\epsilon S_1^{n_B}}.$$
Moreover, if $z\in \Herm{n_An_B}$, then
\begin{equation}
\|z\|_{S_1^{n_An_B}}\leq 2\sqrt{2}\min\{n_A,n_B\}\|z\|_{S_1^{n_A,\mathrm{sa}}\otimes_\epsilon S_1^{n_B,\mathrm{sa}}}.
\label{key_bound}
\end{equation}
\end{prop}

We did not attempt to optimize the constant $2\sqrt{2}$ above. In particular, at the time of writing we do not know if one can obtain a constant of $1$ on the right-hand side of~\eqref{key_bound}. We suspect that to do so one may need some substantially new ideas. However, it is clear from Equation~\eqref{optimal est SEP-LOCC} that the scaling of $\min\{n_A,n_B\}$ in the previous proposition is optimal (this can also be easily proved by using random matrices).

The proof of Proposition~\ref{key prop} is based on the noncommutative Grothendieck's theorem~\cite{Pisier78, Haagerup85}. We use here the last version proved in~\cite{Haagerup85}, which provides us with the optimal constant.
\begin{theorem}\label{NCGT_Theorem}
Let $V:A\times B\rightarrow \C$ be a bounded bilinear form on a pair of C$^*$-algebras $A$ and $B$. Then, there exist two states $\varphi_1$, $\varphi_2$ on $A$ and two states $\psi_1$, $\psi_2$ on $B$ such that $$|V(x,y)|\leq \|V\|\left(\varphi_1(x^*x)+\varphi_2(xx^*)\right)^{\frac{1}{2}}\left(\psi_1(y^*y)+\psi_2(yy^*)\right)^{\frac{1}{2}}$$for all $x\in A$ and $y\in B$.
\end{theorem}

We will apply the previous theorem for the particular case $A=\gM_{n_A}$ and $B=\gM_{n_B}$, where the action of the state $\rho$ on an element $x$ is given by $\tr(\rho x)$. However, we will keep the notation $\rho(x)$ for the simplicity of writing. For the proof of Proposition~\ref{key prop} we will borrow some ideas from the proof of~\cite[Theorem 1]{Blecher89}.

\begin{proof}[Proof of Proposition~\ref{key prop}]
Without loss of generality we can assume that $n_A\leq n_B$. 

Given an element $z\in \gM_{n_An_B}$, let us consider the associated bilinear form $V_z:\gM_{n_A}\times \gM_{n_B}\rightarrow \C$, according to Equation~\eqref{duality_Bil_tensor}, such that
\begin{align}\label{equivzV}
    \|z\|_{S_1^{n_A}\otimes_{\epsilon} S_1^{n_B}}=\|V_z:S_\infty^{n_A}\times S_\infty^{n_B}\rightarrow \C\|.
\end{align}
More precisely, if $z=\sum_{i}x_i\otimes y_i$, $V_z$ is defined  as $V_z(x,y)=\sum_i \tr(x_ix) \tr(y_i y)$ for every pair $(x,y)\in A\times B$.

It is clear that $$\|z\|_{S_1^{n_An_B}}=\max_{\|U\|_{S_\infty^{n_An_B}}=1}\tr(zU),$$where the standard $\sup$ has been replaced by a $\max$ because the spaces have finite dimension and one can assume the maximum to be attained at a unitary matrix $U$.

Let us write $z=\sum_{i,j=1}^{n_A} E_{i,j}\otimes z_{i,j}\in \gM_{n_A}(\gM_{n_B})$ and $U=\sum_{i,j=1}^{n_A} E_{i,j}\otimes U_{i,j}\in \gM_{n_A}(\gM_{n_B})$, where $E_{i,j}=|i\rangle\langle j|$ is the matrix whose $(i,j)$-th entry is $1$ and the remaining entries are all zero. Then, $$\|z\|_{S_1^{n_An_B}}=\tr(zU)=\sum_{i,j=1}^{n_A}\tr(z_{i,j}U_{j,i})=\sum_{i,j=1}^{n_A}V_z(E_{i,j}, U_{i,j}).$$  Hence, we have  
\begin{align*}
\|z\|_{S_1^{n_An_B}} &\leq \sum_{i,j=1}^{n_A}|V_z(E_{i,j}, U_{i,j})|\\
&\leq  \|V_z\|\sum_{i,j=1}^{n_A}(\varphi_1(E_{i,j}^*E_{i,j})+\varphi_2(E_{i,j}E_{i,j}^*))^{\frac{1}{2}}(\psi_1(U_{i,j}^*U_{i,j})+\psi_2(U_{i,j}U_{i,j}^*))^{\frac{1}{2}}\\
&\leq \|V_z\|\left(\sum_{i,j=1}^{n_A}\varphi_1(E_{i,j}^*E_{i,j})+\varphi_2(E_{i,j}E_{i,j}^*)\right)^{\frac{1}{2}}\left(\sum_{i,j=1}^{n_A}\psi_1(U_{i,j}^*U_{i,j})+\psi_2(U_{i,j}U_{i,j}^*)\right)^{\frac{1}{2}},
\end{align*}where we have used the triangle inequality in the first inequality, Theorem~\ref{NCGT_Theorem} in the second one  and the third inequality follows from Cauchy-Schwarz. 

Now, by the linearity of the $\varphi_i$'s and the $\psi_i$'s, the previous expression can be written as

\begin{align*}
&\|V_z\|\left(\varphi_1\left(\sum_{i,j=1}^{n_A}E_{i,j}^*E_{i,j}\right)+\varphi_2\left(\sum_{i,j=1}^{n_A}E_{i,j}E_{i,j}^*\right)\right)^{\frac{1}{2}}\left(\psi_1\left(\sum_{i,j=1}^{n_A}U_{i,j}^*U_{i,j}\right)+\psi_2\left(\sum_{i,j=1}^{n_A}U_{i,j}U_{i,j}^*\right)\right)^{\frac{1}{2}}\\&\leq \|V_z\|\left(\left\|\sum_{i,j=1}^{n_A}E_{i,j}^*E_{i,j}\right\|+\left\|\sum_{i,j=1}^{n_A}E_{i,j}E_{i,j}^*\right\|\right)^{\frac{1}{2}}\left(\left\|\sum_{i,j=1}^{n_A}U_{i,j}^*U_{i,j}\right\|+\left\|\sum_{i,j=1}^{n_A}U_{i,j}U_{i,j}^*\right\|\right)^{\frac{1}{2}}.
\end{align*}

It is easy to see that 
\begin{align}\label{estimate1}
\sum_{i,j=1}^{n_A}E_{i,j}^*E_{i,j}=\sum_{i,j=1}^{n_A}E_{i,j}E_{i,j}^*=n_A \uno_A
\end{align} and 
\begin{align}\label{estimate2}
\sum_{i,j=1}^{n_A}U_{i,j}^*U_{i,j}=\sum_{i,j=1}^{n_A}U_{i,j}U_{i,j}^*=n_A \uno_B.
\end{align}

Indeed, checking Equation~\eqref{estimate1} is straightforward and Equation~\eqref{estimate2} can be shown by writing $$\sum_{i,j=1}^{n_A}U_{i,j}U_{i,j}^*=(tr_A\otimes \uno_B)\left(\sum_{i,j,k=1}^{n_A}E_{i,k}\otimes U_{i,j}U_{k,j}^*\right)=(tr_A\otimes \uno_B)(UU^*)=(tr_A\otimes \uno_B)(\uno_{AB})=n_A\uno_B$$and analogously for $\sum_{i,j=1}^{n_A}U_{i,j}^*U_{i,j}$.

Since according to Equation~\eqref{equivzV} we have $\|V_z\|=\|z\|_{S_1^{n_A}\otimes_\epsilon S_1^{n_B}}$, the previous estimates lead to the upper bound
\begin{align*}
\|z\|_{S_1^{n_An_B}}\leq 2n_A \|z\|_{S_1^{n_A}\otimes_\epsilon S_1^{n_B}}
\end{align*} as we wanted.

The second part of the statement follows from the estimate $$\|z\|_{S_1^{n_A}\otimes_\epsilon S_1^{n_B}}\leq \sqrt{2}\|z\|_{S_1^{n_A,\mathrm{sa}}\otimes_\epsilon S_1^{n_B,\mathrm{sa}}},$$ proved in~\cite[Claim 4.7]{ReVi15}. Indeed, with this estimate at hand, one can write
$$\|z\|_{S_1^{n_An_B,\mathrm{sa}}}\leq \|z\|_{S_1^{n_An_B}}\leq 2n_A\|z\|_{S_1^{n_A}\otimes_\epsilon S_1^{n_B}}\leq 2\sqrt{2}n_A\|z\|_{S_1^{n_A,\mathrm{sa}}\otimes_\epsilon S_1^{n_B,\mathrm{sa}}}.$$
\end{proof}

\section{Emergent classicality in quantum channels }

Throughout this section, we show how to apply our result to give a tighter estimate on the emergence of the quantum-to-classical transition within the context of quantum Darwinism~\cite{Zurek2009, Brandao2015, Knott2018, Eugenia-Darwin, Ranard2020}. With this framework one intends to explore the appearance of the classical notion of objectivity from the quantum world. Consider a finite-dimensional quantum system $A$ that is subjected to a generic dynamics, modelled by a quantum channel $\Lambda$ with composite output system $B\coloneqq B_1\ldots B_n$. Mathematically, this can be modelled by a completely positive trace preserving linear map $\Lambda : \Herm{n_A} \to \Herm{n_B}$, where $n_A,n_B$ are the Hilbert space dimensions of $A$ and $B$, respectively. The $B_i$ systems represent the various \emph{observers} that are gaining information on $A$ through some complicated and partly uncontrollable interaction with it, e.g.\ by employing a measurement apparatus.

With this idea in mind, the archetypal example of a quantum channel for us is the broadcast of a measurement outcome: given a measurement $\left(e_i^A\right)_{i\in I}$ on $A$ and a collection of states $\left(\rho_i^{B}\right)_{i\in I}$ on $B= B_1\ldots B_n$, a measure-and-broadcast channel is given by
\begin{equation}
    \mathcal{E}_{A\to B}(x) \coloneqq \sum_i \tr (e_i^A x)\, \rho_i^B .
    \label{measure-and-broadcast}
\end{equation}
Quantitative statements concerning the phenomenon of quantum Darwinism dictate that if the number $n$ of observers is large, then there will a \emph{single} measurement on $A$ such that for \emph{most} observers the effective channel is well approximated by the broadcast of the outcome of that measurement (see e.g.~\cite[Theorem~1]{Brandao2015},~\cite[Theorem~4]{Eugenia-Darwin}, and~\cite[Theorem~1]{Ranard2020}). This is commonly referred to as \emph{objectivity of observables}.

Clearly, in order for the above statement to make sense, we need to introduce a measure of how close two quantum channels are. A suitable way to do so is via the notion of diamond norm distance. For a linear map $\Gamma: \Herm{n_A} \to \Herm{n_B}$, one starts by defining
\begin{equation}
    \left\|\Gamma\right\|_{1\to 1} \coloneqq \sup_{x\neq 0} \frac{\left\|\Gamma(x)\right\|_1}{\|x\|_1} ,
\end{equation}
where the supremum is over all selfadjoint $n_A\times n_A$ matrices $x$. This quantity is not particularly useful by itself because it is not submultiplicative; however, we can use it to construct a much better behaved object, the so-called diamond norm, given by~\cite{Aharonov1998}
\begin{equation}
    \left\|\Gamma\right\|_\diamond \coloneqq \sup_{m\in \N} \left\| \Gamma\otimes \id_m\right\|_{1\to 1} ,
\end{equation}
where $\id_m$ denotes the identity on the space $\Herm{m}$ of $m\times m$ Hermitian matrices. Among its many appealing features, the diamond norm distance $\left\|\Lambda_1 - \Lambda_2\right\|_\diamond$ between two quantum channels has an operational interpretation in terms of a channel discrimination task~\cite{Sacchi2005}.

We can now give the following improved version of~\cite[Theorem~1]{Ranard2020}:

\begin{theorem}\label{Thm-convergence}
Consider a quantum channel $\Lambda_{A\to B}$ with output system $B=B_1\ldots B_n$. For outputs subsets $R\subset \{B_1,\cdots, B_n\}$, let $\Lambda_{A \to R} \coloneqq \tr_{B\setminus \bar{R}}\circ \Lambda_{A\to B}$ denote the reduced channel onto $R$, obtained by tracing out the complement $\bar{R}$. Then, for any $q,r \in \{1,\cdots, n\}$ there exists a POVM $\left(M_\alpha^A\right)_\alpha$ on $A$ and an ``excluded'' output subset $Q\subset \{B_1,\cdots, B_n\}$ of size $|Q|=q$ with the following property: for all output subsets $R$ of size $|R|=r$ disjoint from $Q$ there exist states $\left(\sigma_\alpha^A\right)_\alpha$ such that the associated measure-and-broadcast channel $\mathcal{E}_{A\to R}(x) \coloneqq \sum_\alpha \tr(e_i^A x) \sigma_\alpha^R$ satisfies
\begin{align}\label{omega-factor}
\|\Lambda_{A\to R}-\mathcal{E}_{A \to R}\|_{\diamond}\leq d_A \Omega_{d_A,d_R} \sqrt{2\ln (d_A)\frac{|R|}{|Q|}},
\end{align}
where $d_A, d_R$ denote the Hilbert space dimensions of $A,R$, respectively, and
\begin{equation}
    \Omega_{d_A, d_R} \coloneqq \left\{ \begin{array}{ll} \min\left\{ 4, 2d_R-1 \right\} & \text{if $d_A=2$,} \\[1ex] \min\left\{ 2\sqrt2 d_A,\, 2d_R-1\right\} & \text{if $d_A\geq 3$.} \end{array} \right.
\end{equation}
\end{theorem}
\begin{proof}
When $d_A=2$ the statement is a reformulation of that of Ranard et al.~\cite[Theorem~1]{Ranard2020}. When $d_A\geq 3$, thanks to our Theorem~\ref{main hidden} we can improve the estimate in~\cite[Eq.~(30)]{Ranard2020} to
\begin{align*}
\|L_{AB}\|_1 &\leq \mathcal{R}(\lo)\, \|L_{AB}\|_\lo \leq \mathcal{R}(\lo)\, \|L_{AB}\|_{\mathrm{LOCC}_\leftarrow} \leq 2\sqrt 2 \min\left\{ d_A, d_B\right\} \|L_{AB}\|_{\mathrm{LOCC}_\leftarrow} .
\end{align*}
This is valid for all Hermitian operators on the bipartite quantum system $AB$. On the right-hand side, Ranard et al.\ have instead the dimensional factor
\begin{equation*}
\Omega_{d_A, d_B}^{\text{Ranard}} \coloneqq \min\left\{d_A^2,\, 4d_A^{3/2},\, 4d_B^{3/2},\, \sqrt{153\, d_Ad_B},\, 2d_B-1 \right\} .
\end{equation*}
For the case at hand ($d_A\geq 3$), it is not difficult to verify that
\begin{equation*}
\Omega_{d_A, d_B} \coloneqq \min\left\{ 2\sqrt2\min\{d_A,d_B\},\, \Omega_{d_A, d_B}^\text{Ranard}\right\} = \min\left\{ 2\sqrt2 d_A,\, 2d_B-1\right\} .
\end{equation*}
The rest of the proof is as in~\cite[Theorem~1]{Ranard2020}.
\end{proof}

The above result tells us that the objectivity of observables emerges before than predicted by Ranard et al., i.e., for a smaller number of observers. The improvement we have given is important for at least two reasons. First, it betters the scaling of the dimensional factor in $d_A$, which is perhaps the most important parameter here. To see why this is so, consider that $d_R$ is likely to be very large in applications, even when $r=1$, because the systems $B_i$ are typically macroscopic --- it is indeed not by chance that they are referred to as ``observers''. Second, it is in some sense the optimal dimensional factor that can be obtained with the proof techniques of~\cite{Brandao2015, Eugenia-Darwin, Ranard2020} --- in fact, it makes~\cite[Lemma~1]{Ranard2020} tight up to constants.

\section{Quantum XOR games}

The great relevance of classical XOR games in both computer science and quantum information motivated the authors in~\cite{ReVi15}  to introduce quantum XOR games. A bipartite quantum XOR game is described by means of a family of bipartite quantum states $(\rho_x)_{x=1}^N$ of local dimensions $n_A$ and $n_B$, a family of signs $c=(c_x)_{x=1}^N\in \{-1,1\}^N$ and a probability distribution $p=(p_x)_x$ on $\{1,\cdots, N\}$.  In order to understand the game, we can think of two players (spatially separated) and a referee. The game starts with the referee choosing one of the states $\rho_x$ according to the probability distribution $p$. Then, the referee sends half of the state to Alice and half of the state to Bob. After receiving the states, Alice and Bob must answer an output, $a=\pm 1$ in the case of Alice and $b=\pm 1$ in the case of Bob. Then, the players win the game if $ab=c_x$. There exists a very close connection between bipartite quantum XOR games and data hiding, as we will see below.

Obviously, the winning probability of the game will strongly depend on the form of the strategies. This form will be determined by the \emph{resources} allowed to Alice and Bob to play the game. In fact, when working with XOR games, it is very common to study the \emph{bias} of the game, $\beta (G)=\mathbf{P}_{win}(G)-1/2$ or, equivalently,
\begin{align*}
\mathbf{P}_{win}(G)-\mathbf{P}_{lose}(G),
\end{align*}rather than the $\mathbf{P}_{win}(G)$ itself.

If the players are allowed to perform any global quantum measurement, it is not difficult to see that the supremum of the bias of the game $G$ under these kinds of strategies can be written as
\begin{align*}
\beta_{\all}(G)=\sup\{\tr(X G): \, \|X\|_{B_{S_\infty^{n_An_B,\mathrm{sa}}}}\leq 1\}=\|G\|_{S_1^{n_An_B, sa}}, 
\end{align*}where \begin{align}\label{XOR quantum-operator}
G=\sum_{x=1}^Nc_xp_x \rho_x \in M_{n_An_B}^{\mathrm{sa}}.
\end{align} 

Another kind of strategies are those  where Alice and Bob must answer independently. These strategies are usually called \emph{product strategies}~\cite{ReVi15}. In this case the supremum of the bias of the game $G$ under these kinds of strategies is given~\cite[Definition 4.3]{ReVi15} by 
\begin{align*}
\beta(G)=\sup\{tr\big((A\otimes B)  G\big): \, \|A\|_{S_\infty^{n_A,\mathrm{sa}}}\leq 1, \|B\|_{S_\infty^{n_B,\mathrm{sa}}}\leq 1\}=\|G\|_{S_1^{n_A, sa}\otimes_\epsilon S_1^{n_B, sa}}.
\end{align*}

Hence, a direct application of Proposition~\ref{key prop} is the following upper bound. 
\begin{theorem}\label{thm: QXOR}
Given a bipartite quantum XOR game with local dimensions $n_A$ and $n_B$ and whose associated tensor is $G=\sum_{x=1}^Nc_xp_x \rho_x$, we have $$\beta_{\all}(G)\leq 2\sqrt{2}\min\{n_A,n_B\}\, \beta(G).$$
\end{theorem}

It is worth mentioning that, in contrast to Theorem~\ref{main hidden} for the distinguishability problem, in this case Theorem~\ref{thm: QXOR} is equivalent to Proposition~\ref{key prop}. Indeed, while in Section~\ref{sec:Distinguishability} the $\epsilon$ norm was used to lower bound the norm $\|\cdot\|_{\lo}$ via Lemma~\ref{lemma norms}, in the case of quantum XOR game, the $\epsilon$ norm precisely describes the products bias $\beta(G)$. Again, the optimality of the estimate in Theorem~\ref{thm: QXOR} can be deduced from the existing results for LOCC.

As it already happened in the data hiding problem, Theorem~\ref{thm: QXOR} can be understood as the ultimate upper bound for general strategies in the sense that product strategies are the most basic ones. Hence, the upper bound provided in Theorem~\ref{thm: QXOR} applies to any kind of strategies.

\section*{Acknowledgements}

Willian Corr\^{e}a was supported by S\~{a}o Paulo Research Foundation (FAPESP), grants 2016/25574-8, 2018/03765-1 and 2019/09205-0. Ludovico Lami acknowledges financial support from the European Research Council under the Starting Grant GQCOP (Grant no.~637352), from the Foundational Questions Institute under the grant FQXi-RFP-IPW-1907, and from the Alexander von Humboldt Foundation. Carlos Palazuelos is partially supported by Spanish MINECO through Grant No.~MTM2017-88385-P, by the Comunidad de Madrid through grant QUITEMAD-CM P2018/TCS4342 and by SEV-2015-0554-16-3.

\end{document}